\documentclass[a4paper,english,nostix]{birthday}
\usepackage{amsfonts}
\usepackage{rotating}
\usepackage{chapterbib}
\usepackage{graphicx, color}
\usepackage[latin1]{inputenc}
\usepackage[english]{babel}
\usepackage{amssymb}
\usepackage{amsmath, eurosym, layout, hyperref} 
\usepackage{times}
\usepackage{amsthm}

\newcommand{\IR}{\mathbb{R}}

\newcommand{\IN}{\mathbb{N}}

\newcommand{\IE}{\mathbb{E}}
\newcommand{\IP}{\mathbb{P}}

\newtheorem{theorem}{Theorem}[section]
\newtheorem{proposition}[theorem]{Proposition}
\newtheorem{lemma}[theorem]{Lemma}
\newtheorem{corollary}[theorem]{Corollary}

\theoremstyle{definition}
\newtheorem{remark}[theorem]{Remark}
\newtheorem{remarks}[theorem]{Remarks}

\newtheorem{definition}[theorem]{Definition}

\newtheorem{examples}[theorem]{Examples}

\newcommand{\bew}{\noindent\textbf{Proof:}\quad}
\newcommand{\ebew}{\hfill\qed\\}

\newsavebox{\prN}

\renewenvironment{proof}{\bew}{\ebew}

\newenvironment{balign*}[1][12pt]{\setlength{\jot}{#1}\nonumber\align}{\endalign}

\renewcommand{\epsilon}{\varepsilon}

\newcommand\ueber[3][]{\genfrac{}{}{0pt}{#1}{#2}{#3}}

\begin{document}

\title{Sixty Years of Moments for Random Matrices\\[5mm]}
\author{Werner Kirsch and Thomas Kriecherbauer}
\contact[werner.kirsch@fernuni-hagen.de]{Fakult\"{a}t f\"{u}r Mathematik und Informatik, FernUniversit\"{a}t Hagen, Germany }
\contact[thomas.kriecherbauer@uni-bayreuth.de]{Mathematisches Institut, Universit\"{a}t Bayreuth, Germany}
\begin{abstract}

This is an elementary review, aimed at non-specialists, of results that have been obtained for the limiting distribution of eigenvalues and for the
operator norms of real symmetric random matrices via the method of moments. This method goes back to a
remarkable argument of Eugen Wigner some sixty years ago which works best for independent matrix entries,
as far as symmetry permits, that are all centered and have the same variance. We then discuss variations of this classical
result for ensembles for which the variance may depend on the distance of the matrix entry to the diagonal, including in
particular the case of band random matrices, and/or for which the required independence of the matrix entries is replaced by
some weaker condition. This includes results on ensembles with entries from Curie-Weiss random variables
or from sequences of exchangeable random variables that have been obtained quite recently.

\end{abstract}
\vspace{5mm}
\begin{center}
\emph{Dedicated to Helge Holden on the occasion of his 60$^{\text{\,th}}$ birthday}
\end{center}

\section{Introduction}\label{sec:Intro}
Approximately at the time when Helge Holden was born the physicist Eugene Wigner presented a result in \cite{Wigner1}
that may be considered to be the starting signal for an extremely fruitful line of investigations creating the now
ample realm of random matrices. The reader may consult the handbook \cite{handbookRMT} to obtain an impression
of the richness of the field. Its ongoing briskness is well documented by over 700 publications
listed in MathSciNet after the print of \cite{handbookRMT} in 2011.

In view of later developments that often use heavy machinery to provide very detailed knowledge about specific
spectral statistics, Wigner's observation impresses by its simplicity and fine combinatorics. For certain matrix ensembles,
which in various generalizations are nowadays called {\em Wigner ensembles},
he was able to determine the limiting density of eigenvalues by the {\em moment method}.
More precisely, he computed the expectations of all moments of the {\em eigenvalue distribution measures}
in the limit of matrix dimensions tending to infinity.
Furthermore, he observed that these limits agree with the moments of the {\em semicircle distribution}
thus proving the semicircle law that bears his name (see Sections \ref{ssec:scl-setup} and \ref{ssec:scl-wigner}
for definitions of the phrases in italics).

It is quite remarkable that the moment method continues to provide new insights into the distribution of random eigenvalues.
With this article we take the reader on a tour that starts with Wigner's discovery and ends with the description of
recent results, some yet unpublished. Along the way we try to explain a few developments in more detail while briefly pointing at others.

The first application of the moment method to the analysis of random eigenvalues appears almost accidental.
In an effort to understand ``the wave functions of quantum mechanical systems which are assumed to be so complicated
that statistical considerations can be applied to them'' Wigner introduces in \cite{Wigner1} three types of ensembles of
``real symmetric matrices of high dimensionality''. Although he considers his results not satisfactory from a physical point of view,
he expresses the hope that ``the calculation which follows may have some independent interest''. Moreover, the reader learns that
one of the three models considered just serves ``as an intermediate step''.
And it is only this auxiliary ensemble that we would now call a Wigner ensemble.
Wigner names it the ``random sign symmetric matrix'' by which he understands $(2N+1)\times (2N+1)$ matrices
for which the diagonal elements are zero and
``non diagonal elements $v_{ik} = v_{ki} = \pm v$ have all the same absolute value but random signs".

In the short note \cite{Wigner2} that appeared a few years later, Wigner remarks that the arguments of
\cite{Wigner1} show the semicircle law for a much larger class of real symmetric ensembles. He observes that,
except for technical assumptions, two features of the model were essential for his proof:
Firstly, stochastic independence of the matrix entries (as far as the symmetry permits) and,
secondly, that all (or at least most) matrix entries are centered and have the same variance.

In Section \ref{ssec:scl-wigner} we present Wigner's proof with enough detail to make
the significance of these two assumptions apparent. The remaining sections are then devoted
to the discussion of results where at least one of these essential assumptions are weakened.

In Section \ref{ssec:scl-band} independence and centeredness of the matrix entries are kept.
However, we allow the variances to vary as a function of the distance to the diagonal.
The most prominent examples in this class are band matrices and we discuss them in detail.

A first step to loosen the assumption of independence is presented in
Section \ref{ssec:scl-sparse}. Its central result provides conditions on the number and location of
matrix entries that may be dependent without affecting the validity of Wigner's reasoning. We call such
dependence structures  \emph{sparse}. Sparse dependence structures
appear for example in certain types of block random matrices that are used in modelling disordered systems in
mesoscopic physics (see e.\,g. \cite{AltZ})     .

In the last three sections \ref{ssec:scl-deccor}--\ref{ssec:scl-ee} we report on results for ensembles
with a dependence structure that is not sparse. This is largely uncharted territory.
However, in recent years a number of special cases were analyzed using the method of moments.
They show interesting phenomena that should be explored further. We devide the models into three groups.
In the first group the correlations decay to $0$
as the distance of the matrix entries becomes large in some prescribed metric.
Then we look at those ensembles for which the entries are drawn from
Curie-Weiss random variables. Here the correlations have no spatial decay but decay for supercritical temperatures as the matrix
dimension becomes large. Finally, we pick the matrix entries from an infinite sequence of exchangeable random variables.
Here the correlations between matrix entries depend neither on their locations nor on the size of the matrix.

We close this introduction by stating what is not contained in this survey. One of the striking features of random matrix theory is
the observation that local statistics of the eigenvalues obey universal laws that, somewhat surprisingly, have also arisen in
certain combinatorial problems, in some models from statistical mechanics and even in the distribution of the non-trivial zeros of
zeta-functions. By local statistics we mean statistics after local but deterministic rescaling so that the spacings between
neighboring eigenvalues are of order $1$. Examples are the statistics of spacings or the distribution of extremal eigenvalues.
Such results were first obtained for Gaussian ensembles, i.e.~Wigner ensembles with normally distributed entries.
In this special case it is possible to derive an explicit formula for the joint distribution of eigenvalues that can then be analyzed
using the method of orthogonal polynomials. In the Gaussian case this requires detailed asymptotic formulas for
Hermite polynomials of large degree that had already been derived in the beginning of the twentieth century.
The first step
to prove universality beyond Gaussian ensembles was then taken about twenty years ago within the
class of ensembles that are invariant under change of orthonormal bases.
For such ensembles the eigenvectors are distributed according to Haar measure,
the joint distribution of eigenvalues is still explicit, and the method of orthogonal polynomials works,
albeit they generally do not belong to the well studied families of classical orthogonal polynomials
(see for example \cite{Deift,DeiftG,PasturS} and references therein).
It is only seven years ago that universality results for local statistics became available for Wigner matrices (see e.\,g. \cite{BGKnowles,ErdosKYY,GoetzeNTT,TaoV} and references therein).
Since all of these results do not use the moment method, we will not discuss them in this paper.

There is one notable exception to what has just been said. The distribution of extremal eigenvalues
(and consequently of the operator norm) can and has been investigated for Wigner ensembles on the local level,
using the method of moments \cite{Sosh}. However, this requires quite substantial extensions of the ideas that we explain
and goes way beyond the scope of this paper. We therefore only state weaker results that might be considered
as laws of large numbers for the operator norm and that can be
proved with much less effort. Nevertheless, we do not discuss their proofs either and refer the reader to Section 2.3 of the
textbook \cite{Tao}.

Finally, we mention that the moment method can also be applied to complex Hermitian matrices and
to sample covariance matrices (also known as Wishart ensembles),
but in the present article we always restrict ourselves to the case of real symmetric matrices to keep the presentation as elementary as possible.

\section{Setup}\label{ssec:scl-setup}
We begin by setting the scene and fixing some notation.
\begin{definition}
A (real symmetric) \emph{matrix ensemble} is a family

$X_N (i,j), i,j = 1, \ldots, N, N \in \IN$ of real valued random variables on a probability space $(\Omega, \mathcal{F}, \IP)$ such that $X_N (i,j) = X_N (j,i)$.
We then denote by $X_N$ the corresponding
$N \times N$-matrix, i.e.

\begin{equation}
X_N =
\begin{pmatrix}
 X_N (1,1) & X_N (1,2) & \cdots & X_N (1,N)\\
 X_N (2,1) & X_N (2,2) & \cdots & X_N (2,N)\\
 \vdots    & \vdots    &        & \vdots   \\
 X_N (N,1) & X_N (N,2) & \cdots & X_N (N,N)
\end{pmatrix}
\end{equation}
\end{definition}
Since we deal exclusively with real symmetric matrices by 'matrix ensemble' we always mean a real symmetric one.
\begin{definition}\label{def:Wigner}
A (real symmetric) matrix ensemble is called \emph{independent} if for each $N \in \IN$ the random variables $X_N (i,j), 1 \leq i \leq j \leq N$ are independent. It is called \emph{identically distributed}, if all $X_N (i,j)$ have the same distribution.\\
An independent and identically distributed matrix ensemble $X_N$ is called a \emph{Wigner ensemble if $\IE (X_N (i,j)) = 0$ and $\IE (X_N (i,j)^2) = 1$ }.
\end{definition}
By a slight abuse of language we use the phrase '$X_N$ is a Wigner matrix' to indicate that the family $X_N$ of random (symmetric) matrices form a Wigner ensemble.
Some authors allow for Wigner ensembles a probability distribution for the diagonal elements which differs from the distribution for the nondiagonal entries.
\begin{definition}\label{def:moments}
The \emph{$k^{th}$ moment} of a random variable $X$ is the expectation $\IE (X^k)$.
We say that all moments of $X$ exist, if $\IE (\vert X \vert^k) < \infty$ for all $k \in \IN$.
\end{definition}
Unless stated otherwise we always assume that all random variables occurring in this text have all moments existing.

For any symmetric $N \times N$-matrix $M$ we denote the eigenvalues of $M$ by $\lambda_j (M)$. We order these eigenvalues such that
\begin{equation*}
\lambda_1 (M) \leq \lambda_2 (M) \leq \ldots \leq \lambda_N (M)
\end{equation*}
where degenerate eigenvalues are repeated according to their multiplicity.

The eigenvalue distribution measure $\nu_N$ of $M$ is defined by
\begin{align*}
\nu_N (A) & = \frac{1}{N} \;\Big| \{j \mid \lambda_j (M) \in A \} \Big\vert\\
& = \frac{1}{N} \sum_{j=1}^N \delta_{\lambda_j (M)} (A)
\end{align*}
where $\vert B \vert$ denotes the number of points in $B, N$ - as above - is the dimension of the matrix $M$, $A$ is a Borel-subset of $\IR$ and $\delta_a$ is the Dirac measure in $a$, i.e.
\begin{equation}\label{eq:delta}
\delta_a (A) =
\begin{cases}
1 & \text{~ if ~} a \in A\\
0 & \text{~  otherwise} \,.
\end{cases}
\end{equation}

It turns out that for a Wigner matrix $X_N$ the eigenvalue distribution measure $\nu_N$ of $X_N$ has no chance to converge as $N\to\infty$ as the following back-of-the-envelope calculations show.

We have
\begin{align}
\int \lambda^2 ~  d \nu_N (\lambda) & = \frac{1}{N} \sum_{\ell=1}^N \lambda_{\ell} (X_N)^2\notag\\
& = \frac{1}{N} \text{~  tr } X_N^2 \label{eq:HS0}
\end{align}

If the $N \times N$-matrix $X_N$ has entries $\pm 1$ (random or not), then \eqref{eq:HS0} shows
\begin{equation}
\int \lambda^2 ~  d \nu_N (\lambda) = N \label{eq:HS1}
\end{equation}

and if the $X_N$ are random matrices with $\IE (X_N (i,j)^2) = 1$ we get
\begin{equation}\label{eq:HS2}
\IE \Big(\int \lambda^2 ~  d \nu_N (\lambda) \Big) = N \,.
\end{equation}

This shows that (at least the second moment of) the eigenvalue distribution measure of $X_N$ is divergent.

Moreover, the same calculation suggests that the eigenvalue distribution measure of the normalized matrices
\begin{equation*}
M_N = \frac{1}{\sqrt{N}} X_N
\end{equation*}
might converge as for $M_N$
\begin{equation*}
\IE \int \lambda^2 ~  d \nu_N = \frac{1}{N} \text{~  tr ~} M_N^2 = 1 \,.
\end{equation*}
As we shall see below, this is indeed the case not only for Wigner ensembles, but for a huge class of random matrices.

A similar reasoning applies to the operator norm of a matrix ensemble $X_N$:
\begin{equation}\label{eq:defNorm}
\Vert X_N \Vert = \max\, \Big\{ \vert \lambda_1 (X_N) \vert, \vert \lambda_N (X_N) \vert \Big\} \,.
\end{equation}
Since for any real symmetric $N \times N$-matrix $M$:
\begin{equation}\label{eq:HS3}
\frac{1}{N} \text{~  tr ~} M^2 = \frac{1}{N} \sum_{\ell =1}^N \lambda_{\ell} (M)^2 \leq \Vert M \Vert^2 \leq \text{~  tr ~} M^2
\end{equation}
a matrix $M$ with $\pm 1$-entries satisfies
\begin{equation*}
\sqrt{N} \leq ||M|| \leq N
\end{equation*}
and similarly for $\IE (X_N (i,j)^2) = 1$
\begin{equation*}
\sqrt{N} \leq \IE (||X_N||^2)^{\frac{1}{2}} \leq N \,.
\end{equation*}
Again, one is lead to look at the norm of $M_N = \frac{1}{\sqrt{N}} X_N$.

Indeed, for Wigner ensemble the norm of $M_N$ will stay bounded as $N \rightarrow \infty$, in fact, it will converge to 2.

However, this fact is more subtle than the convergence of $\nu_N$, and so is its proof (cf. Theorem \ref{th:WignerNorm} that was proved by F\"{u}redi and Komlos in \cite{FuerediK}, see also \cite{BaiY}, and \cite{Tao} for a textbook presentation).

To illustrate this, let us look at a particular example within the class considered in \eqref{eq:HS1}, namely the $N \times N$-matrices
\begin{equation}\label{eq:defE}
\mathcal{E}_N (i,j) = 1 \text{~  for all ~} 1 \leq i,j \leq N \,.
\end{equation}
The matrix $\mathcal{E}_N$ can be written as
\begin{equation*}
\mathcal{E}_N = N \cdot P_e
\end{equation*}
where $P_{e}$ is the orthogonal projection onto the vector $e$, with $e (i) = \frac{1}{\sqrt{N}}$ for\goodbreak $i = 1, \ldots , N$.

Consequently $\mathcal{E}_N$ is of rank 1 and

\begin{equation}\label{eq:lambaE}
\lambda_j (\mathcal{E}_N) =
\begin{cases}
N & \text{~  for ~} j = N\\
0 & \text{~  otherwise} \,.
\end{cases}
\end{equation}

Thus we obtain
\begin{equation*}
\Big|\Big|\frac{\mathcal{E}_N}{\sqrt{N}}\Big|\Big| = \sqrt{N} \rightarrow \infty
\end{equation*}
but the eigenvalue distribution function $\nu_N$ of $\mathcal{E}_N/\sqrt{N}$ is given by
\begin{equation*}
\nu_N~=~ \frac{1}{N} (N - 1) \delta_0 + \frac{1}{N} \delta_{\sqrt{N}} ~~\Longrightarrow~~ \delta_0
\end{equation*}
where $\Rightarrow$ means weak convergence (see definition \ref{def:wconv}).

\section{Wigner's Semicircle Law}\label{ssec:scl-wigner}
In this section we present and discuss the classical semicircle law for Wigner ensembles.

So, let $X_N$ be a Wigner ensemble (see Definition \ref{def:Wigner}), set $M_N = \frac{1}{\sqrt{N}} X_N$ and denote the eigenvalue distribution measure of $M_N$ by $\sigma_N$, thus
\begin{equation*}
\sigma_N (A) = \frac{1}{N}\; \Big| \{j\mid \lambda_j (\frac{1}{\sqrt{N}} X_N) \in A \}\Big| \,.
\end{equation*}
The semicircle law, in its original form due to Wigner (\cite{Wigner1}, \cite{Wigner2}), states that $\sigma_N$ converges to the semicircle distribution $\sigma$ given through its Lebesgue density

\begin{equation}\label{eq:defscd}
\sigma (x) =
\begin{cases}
\;\frac{1}{2 \pi} \sqrt{4 - x^2} & \text{~  , for ~} |x| \leq 2\\
\;0 & \text{~  , otherwise.}
\end{cases}
\end{equation}

$\sigma$ describes a semicircle of radius 2 around the origin, hence the name.

So far, we have avoided to explain in which sense $\sigma_N$ converges. This is what we do now.

Let us first look at the convergence of measures on $\IR$.
\begin{definition}
\label{def:wconv}
Suppose $\mu_N$ and $\mu$ are probability measures on $\IR$ (equipped with the Borel $\sigma$-algebra).

We say that $\mu_N$ \emph{converges weakly} to $\mu$, in symbols
\begin{equation*}
\mu_N \Rightarrow \mu
\end{equation*}
if
\begin{equation*}
\int f (x) ~  d \mu_N (x) \rightarrow \int f (x) ~  d \mu (x)
\end{equation*}
for all $f \in C_b (\IR)$, the space of bounded continuous functions.
\end{definition}
If the matrix $X_N$ is random and
\begin{equation*}
\sigma_N = \frac{1}{N} \sum \delta_{\lambda_j} \Big(\frac{X_N}{\sqrt{N}}\Big)
\end{equation*}
is the eigenvalue distribution measure of $\frac{1}{\sqrt{N}} X_N$ then the measure $\sigma_N$ itself is random.

Consequently, we have not only to define in which sense the measures converge (namely weakly), but also how this convergence is meant with respect to randomness, i.e. to the `parameter' $\omega \in \Omega$.
There are various ways to do this.
\begin{definition}\label{def:Tprobwconv}
Let $(\Omega, \mathcal{F}, \IP)$ be a probability space and let $\mu_N^\omega$ and $\mu^\omega$ be random probability measures on $(\IR, B(\IR))$.
\begin{enumerate}
\item[1)] We say that $\mu_N^\omega$ converges to $\mu^\omega$ \emph{weakly in expectation}, if for every

$f \in C_b (\IR)$
\begin{equation}\label{eq:convwe}
\IE \Big( \int f (x) ~  d \mu_N^{\omega} (x) \Big) \rightarrow \IE \Big(\int f (x) ~  d \mu^\omega (x) \Big)
\end{equation}

as $N \rightarrow \infty$.
\item[2)] We say that $\mu_N^\omega$ converges to $\mu^\omega$ \emph{weakly in probability},
if for every $f \in C_b (\IR)$ and any $\epsilon > 0$
\begin{equation*}
\IP \Big(\;\Big| \int f (x) ~ d\mu_N^\omega (x) - \int f (x) ~ d \mu^\omega (x) \Big| > \epsilon \Big) \rightarrow 0
\end{equation*}
as $N \rightarrow \infty$.
\item[3)] We say that $\mu_N^\omega$ converges to $\mu^\omega$ \emph{weakly $\IP$-almost surely} if there is a set $\Omega_0 \subset \Omega$ with $\IP (\Omega_0) = 1$ such that $\mu_N^\omega
 \Rightarrow \mu^\omega$ for all $\omega \in \Omega_0$.
\end{enumerate}
\end{definition}

\begin{theorem}[Semicircle Law]\label{th:Wigner}

Suppose $X_N$ is a Wigner ensemble with

$\IE (|X_N (i,j)|^k) < \infty$ for all $k \in \IN$ and let $\sigma_N$ denote the eigenvalue distribution measure of $M_N = \frac{1}{\sqrt{N}} X_N$ then $\sigma_N$ converges to the semicircle distribution $\sigma$ weakly $\IP$-almost surely.
\end{theorem}
\begin{remarks}

\begin{enumerate}
\item Wigner \cite{Wigner1,Wigner2} proved this theorem for weak convergence in expectation.
\item Grenander \cite{Grenander} showed under the same conditions that the convergence holds weakly in probability.
\item Arnold \cite{Arnold} proved that the convergence is weakly $\IP$-almost surely. He also relaxed the moment condition to
\begin{equation*}
\IE\, \Big(X_N (i,j)^6\Big) ~<~ \infty
\end{equation*}
for $\IP$-almost sure weak convergence and to
\begin{equation*}
\IE\,\Big(X_N (i,j)^4\Big) ~<~ \infty
\end{equation*}
for weak convergence in probability.
\item According to Defition \ref{def:Wigner} the entries in a Wigner ensemble are independent and
identically distributed, so condition $\IE (|X_N (i,j)|^k) < \infty$ actually implies
\begin{equation*}
   \sup_{N,i,j}\,\IE (|X_N (i,j)|^k) < \infty\ .
\end{equation*}
for each $k$.
\end{enumerate}
\end{remarks}
Besides the moment method we discuss in this article there is another important technique
to prove the semicircle law. This is  the Stieltjes transform method originating in \cite{MarchenkoP}, \cite{Pastur72} and \cite{Pastur73}, see also \cite{PasturS} and references given there. Both methods are discussed in \cite{AGZ} and
in \cite{Tao}.

The moment method is based on the observation that the following result is true.
\begin{proposition}
If $\mu_N$ and $\mu$ are probability measures on $\IR$ such that all moments of $\mu_N$ exist and
 \begin{equation}\label{eq:mom}
\int |x|^k ~ d \mu (x) \leq A C^k k!
\end{equation}

for all $k$ and some constants $A, C$, then
\begin{equation*}
\int x^k ~ d \mu_N (x) \rightarrow \int x^k ~ d \mu (x)
\end{equation*}
for all $k \in \IN$ implies that
\begin{equation*}
\mu_N \Rightarrow \mu \,.
\end{equation*}
\end{proposition}
For a proof see for example \cite{Breiman}, \cite{Klenke} or \cite{MoBu}.

Since the semicircle distribution $\sigma$ has compact support, it obviously
satisfies \eqref{eq:mom}. The moments of $\sigma$ are given by:
\begin{align}
   \int x^k \,d\sigma(x)~&=~\left\{
                             \begin{array}{ll}
                               C_{k/2}, & \hbox{if $k$ is even;} \\[2mm]
                               0, & \hbox{if $k$ is odd.}
                             \end{array}
                           \right.\label{eq:momsc}\\
\intertext{where}\qquad C_{\ell}~&=~\frac{1}{\ell+1}\,\binom{2\ell}{\ell}\qquad \label{eq:Catalan}
\end{align}
are the
\emph{Catalan numbers}. (For a concise introduction to Catalan numbers see e.\,g. \cite{Koshy} or \cite{Stanley}.)

The moments of $\sigma_N$ can be expressed through traces of the matrices $X_N$
\begin{align}
& \IE \Big(\int x^k ~ d \sigma_N (x) \Big)~ =~ \frac{1}{N} ~ \IE \Big( \sum_{j=1}^{N} \lambda_j \Big(\frac{X_N}{\sqrt{N}}\Big)^k \Big)
 ~=~ \frac{1}{N}  ~ \IE \Big(\;\text{tr} \big(\frac{X_N}{\sqrt{N}}\big)^k\; \Big)\notag\\
& = \frac{1}{N^{1 + \frac{k}{2}}} \sum_{i_1, \ldots , i_k = 1}^N \IE \Big(X_N (i_1,i_2) \cdot X_N (i_2,i_3) \cdot ~~ \ldots ~~  \cdot X_N (i_k,i_1) \Big)\; .\label{eq:sumXi}
\end{align}

The sum in \eqref{eq:sumXi} contains $N^{k}$ terms. So, at a first glance, the normalizing factor $N^{1+k/2}$ seems too small to compensate
the growth of the sum. Fortunately, many of the summands are zero, as
we shall see later.

For the purpose of bookkeeping it is useful to think of $i_{1}, i_{2},\ldots,i_{k}$ in terms of
a graph.
\begin{definition}\label{def:graph}
The multigraph $\mathcal{G}$ with vertex set
\begin{align}\label{eq:vertices}
\mathcal{V}~&:=~\Big\{i_{1}, i_{2},\ldots,i_{k}\Big\}
\end{align}
and $\ell$ (undirected) edges between $i$ and $j$ if $\{i,j\}$ occurs $\ell$ times
in the sequence
\begin{align}
   \{i_{1},i_{2}\},\{i_{2},i_{3}\},\ldots,\{i_{k},i_{1}\}\label{eq:edges}
\end{align}
is called the multigraph associated with $(i_{1},i_{2},\ldots,i_{k})$.
\end{definition}
\begin{remark}
   The sequence $(i_{1},i_{2},\ldots,i_{k})$ defines a \emph{multi}graph since there may be several edges between the
   vertices $i_{\nu}$.
   \end{remark}
\begin{definition}\label{def:sgraph}
   If $\mathcal{G}$ is a multigraph we define the associated (simple) graph $\widetilde{\mathcal{G}}$ in the following way.
   The set of vertices of $\widetilde{\mathcal{G}}$ is the same as the vertex set of $\mathcal{G}$ and $\widetilde{\mathcal{G}}$
   has a single edges between $i$ and $j$ whenever $\mathcal{G}$ has at least one edge between $i$ and $j$.
\end{definition}
\begin{remark}
   The sequence \eqref{eq:edges} describes not only a multigraph but in addition a closed path through the multigraph which uses each
edge exactly once. Such paths are called \emph{Eulerian circuits}. They occur for example in the famous problem of the
`Seven Bridges of K\"{o}nigsberg' (see e.\,g. \cite{Bollabas}).

The existence of a Eulerian circuit implies in particular that the multigraph is connected.
\end{remark}

Now, we order the sum in \eqref{eq:sumXi} according to the number $|\mathcal{V}|=|\{i_1, \ldots , i_k \}| $
of \emph{different} indices (=vertices)
occurring in the sequence $i_{1}, i_{2},\ldots,i_{k}$.

\begin{align}
   &\sum_{i_1, \ldots , i_k = 1}^N \IE \Big(X_N (i_1,i_2) \cdot X_N (i_2,i_3) \cdot ~~ \ldots ~~  \cdot X_N (i_k,i_1) \Big)\notag\\
   ~=~&\sum_{r=1}^{k}\;\sum_{|\{i_1, \ldots , i_k\}|=r}\;
   \IE \Big(X_N (i_1,i_2) \cdot X_N (i_2,i_3) \cdot ~~ \ldots ~~  \cdot X_N (i_k,i_1) \Big)
\end{align}
The number of index tuples $ (i_1, \ldots , i_k )$ with $|\{i_1, \ldots , i_k \}|=r $ is of order $\mathcal{O}(N^{r})$ and
can be bounded above by
$r^{k}\,N^{r}$. In fact, to choose the $r$ different numbers in $\{1,\ldots,N\} $ we have
less than $N^{r}$ possibilities. Then, to choose which one to put at a given position we have at most
$r$ choices for each of the $k$ positions.

Therefore the sum
\begin{align}
   \sum_{|\{i_1, \ldots , i_k\}|=r}\;
   \IE \Big(X_N (i_1,i_2) \cdot X_N (i_2,i_3) \cdot ~~ \ldots ~~  \cdot X_N (i_k,i_1) \Big)
\end{align}
is of order $\mathcal{O}(N^{r})$ as well. Thus the terms with $r=|\{i_1, \ldots , i_k \}|< 1+k/2 $
in \eqref{eq:sumXi} can be neglected compared to prefactor $N^{-(1+k/2)}$. Consequently
\begin{align}
  \frac{1}{N^{1+k/2}} \sum_{|\{i_1, \ldots , i_k\}|<1+k/2}
   \IE \Big(X_N (i_1,i_2) \cdot ~ \ldots ~ \cdot X_N (i_k,i_1) \Big)~\longrightarrow~0
\end{align}

To handle those terms with $|\{i_1, \ldots , i_k \}|> 1+k/2 $ we need the following two observations.

For comparison with results in Section \ref{ssec:scl-CW} we formulate the first one as a lemma.

\begin{lemma}\label{lem:once}
Whenever an edge $\{i,j\}$ occurs only once in \eqref{eq:edges} then
\begin{equation}\label{eq:EXi}
\IE \Big(X_N (i_1,i_2) \cdot X_N (i_2,i_3) \cdot ~~  \ldots ~~  \cdot X_N (i_k,i_1) \Big)~=~0\,.
\end{equation}
\end{lemma}

This follows from independence and the assumption $\IE \big(X_N (i,j)\big) = 0$.

The second observation is:
\begin{proposition}\label{prop:r-gross}
 If $|\{i_1, \ldots , i_k \}|> 1+k/2 $ there is an edge $\{i,j\}$
 which occurs only once in $ \{i_{1},i_{2}\},\{i_{2},i_{3}\},\ldots,\{i_{k},i_{1}\} $.
\end{proposition}

\begin{proof}
Set $r=|\{i_1, \ldots , i_k \}|$ and denote the distinct elements of $|\{i_1, \ldots , i_k \}|$
by $j_{1},\ldots,j_{r}$.

   To connect the vertices $j_{1}, \ldots, j_{r}$ we need at least $r-1$ edges. To double each
   of these connections we need $2r-2$ edges. So, if we have $k$ edges we need that
   $k\geq 2r-2$ to double each connection. Hence, if $r>1+k/2$, at least one edge occurs only once.
\end{proof}

\begin{remark}\label{rem:tree}
   A similar reasoning as in the proof above shows: If a graph $\mathcal{G}$ with $k$ edges and $k+1$ vertices is connected
   then $\mathcal{G}$ is a \emph{tree}, i.\,e. $\mathcal{G}$ contains no loops. Indeed, if $\mathcal{G}$ contained a loop
   we could remove an edge without destroying the connectedness of the graph. But the new graph would
   have $k-1$ edges and $k+1$
   vertices, so it cannot be connected.
\end{remark}

From Proposition \ref{prop:r-gross} and \eqref{eq:EXi} we learn that
\begin{align}
   \sum_{|\{i_1, \ldots , i_k\}|>1+k/2}\;
   \IE \Big(X_N (i_1,i_2) \cdot X_N (i_2,i_3) \cdot ~~ \ldots ~~  \cdot X_N (i_k,i_1) \Big)~=~0\ .
\end{align}

To summarize, what we proved so far is
\begin{align}
 & \frac{1}{N^{1 + \frac{k}{2}}}\sum_{i_1, \ldots , i_k = 1}^N \IE \Big(X_N (i_1,i_2) \cdot X_N (i_2,i_3) \cdot ~~ \ldots ~~  \cdot X_N (i_k,i_1) \Big)\notag\\
 \approx~&\frac{1}{N^{1 + \frac{k}{2}}} \sum_{\ueber[1]{|\{i_1, \ldots , i_k\}|=1+k/2} {\text{all }\{i,j\} \text{ occur exactly twice}}}\;
   \IE \Big(X_N (i_1,i_2) \cdot X_N (i_2,i_3) \cdot ~~ \ldots ~~  \cdot X_N (i_k,i_1) \Big)
\label{eq:Var1}
\end{align}
Let us set
\begin{align}
   \mathcal{I}_{k}^{(N)}=\Big\{(i_{1},\ldots,i_{k})\in\{1,\ldots,N\}\,\Big|&\; |\{i_1, \ldots , i_k\}|=1+k/2\notag \\&\text{  and all }
\{i,j\}
 \text{  occur exactly twice.}\Big\}
\end{align}

For odd $k$ the set $\mathcal{I}_{k}^{(N)}$ is empty, so the sum \eqref{eq:Var1} is obviously zero.

Due to independence and the assumptions $\IE(X_{N}(i,j))=0$ and $\IE({X_{N}(i,j)}^{2})=1$ we have
\begin{align*}
   \IE \Big(X_N (i_1,i_2) \cdot X_N (i_2,i_3) \cdot ~~ \ldots ~~  \cdot X_N (i_k,i_1) \Big)~=~1
\end{align*}
whenever all $\{i,j\}$ occur exactly twice.

Consequently,
\begin{align}
  \eqref{eq:Var1}~=~\frac{1}{N^{1 + \frac{k}{2}}}\; \big|\mathcal{I}_{k}^{(N)}\big|\,. \label{eq:twice}
\end{align}

For even $k$, let us consider the multigraph $\mathcal{G}$ associated with $(i_{1},\ldots,i_{k})\in\mathcal{I}_{k}^{(N)}$. Since
$\mathcal{G}$ has $1+k/2$ vertices and $k$ double vertices, the corresponding simple graph $\widetilde{\mathcal{G}}$ is a
connected graph with $1+k/2$ vertices and $k/2$ edges. Thus, this $\widetilde{\mathcal{G}}$ is a tree by Remark \ref{rem:tree}.
Moreover the path
$(i_{1},\ldots,i_{k}, i_{1})$ defines an ordering on $\widetilde{\mathcal{G}}$.

The number of ordered trees \cite{Koshy,Stanley} with $\ell $ edges (and hence $\ell+1$ vertices) is known to be the Catalan number $C_{\ell}$ (see \eqref{eq:Catalan}).

Given an (abstract) ordered tree with $\ell=1+k/2$ vertices we find all corresponding paths
$(i_{1},i_{2},\ldots,i_{k},i_{1})$ with $i_{j}\in\{1,\ldots,N\}$
by assigning $1+k/2$ (different) numbers (=indices) from $\{1,\ldots,N\}$
to the vertices of the tree. There are $ \frac{N!}{(N-(1+k/2))!}\approx N^{1+k/2}$ ways to do this.
Thus

\begin{align}
  \IE \Big( \int x^k ~  d \sigma_N (x) \Big)~&\approx~\frac{1}{N^{1 + \frac{k}{2}}}\; \big|\mathcal{I}_{k}^{(N)}\big| \\
   &\rightarrow~
\begin{cases}
C_{\frac{k}{2}} & \text{~  for ~} k \text{~  even}\\
0 & \text{~  for ~} k \text{~  odd}
\end{cases}
\end{align}

and these are the moments of the semicircle distribution $\sigma$ (see  \eqref{eq:momsc}).

For more details on the semicircle law and its proof see \cite{AGZ}, \cite{Tao} or \cite{MoBu}.

From Theorem \ref{th:Wigner} and \eqref{eq:defscd} we conclude that $\liminf ||\frac{1}{\sqrt{N}}X_{N}||\geq 2$ almost surely, since for symmetric $N\times N$-matrices
$A$ the matrix norm $\|A\|$, as an operator on the Euclidian space $\IR^{N}$, satisfies $\|A\|=\max\,\{|\lambda_{1}(A)|,|\lambda_{N}(A)|\}$.

However, Theorem \ref{th:Wigner}
does \emph{not} imply that $\liminf \|\frac{1}{\sqrt{N}}X_{N}\|\leq 2$! Wigner's result does imply that the \emph{majority} of the eigenvalues will be less than $2+\varepsilon$ finally, however some \big(in fact even $o(N)$\big) eigenvalues could be bigger and might even go to $\infty$. In Sections \ref{ssec:scl-band}, \ref{ssec:scl-CW}, and \ref{ssec:scl-ee} we encounter ensembles for which exactly this happens.

However, for Wigner ensembles it is correct that the norm of $\frac{1}{\sqrt{N}}X_{N} $ goes to $2$. This can be shown
by a more sophisticated variant of the moment method.

\begin{theorem}\label{th:WignerNorm}
Suppose $X_N$ is a Wigner ensemble with $\IE \Big( |X_N (i,j)|^k \Big) < \infty$ for all $k \in \IN$ and let
\begin{equation*}
\lambda_N^* = \max \Big\{ |\lambda_1 ( \frac{1}{\sqrt{N}} X_N ) |, |\lambda_N ( \frac{1}{\sqrt{N}} X_N ) | \Big\} = \| \frac{1}{\sqrt{N}}X_{N} \|
\end{equation*}
be the operator norm of $M_N = \frac{1}{\sqrt{N}} X_N$, then
\begin{equation*}
\lambda_N^* \rightarrow 2 \qquad\text{~  as ~} N \rightarrow \infty\quad \text{$\IP$-almost surely.}
\end{equation*}
\end{theorem}
This theorem was proved by F\"{u}redi and Komlos in \cite{FuerediK}, see also \cite{BaiY}.

To prove the semicircle law we considered the $k^{th}$ moment $m_k$ of $\sigma_N$ for \emph{fixed}  $k$ as $N$ goes to infinity. For the norm estimate we need bounds on $m_k$ for $k=k_N$ for a sequence $k_{N}$ which is \emph{growing} with $N$. See \cite[Section 2.3]{Tao} for a pedagogical explanation.

\section{Random Band Matrices}\label{ssec:scl-band}
In a first variation of Wigner's semicircle law we abandon the assumption of identical distribution of
the $X_{N}(i,j)$, by assuming that entries away from a band around the diagonal are zero, while the other entries are still iid, apart from the symmetry $X_N (i,j) = X_N (j,i)$.

More precisely, let $\tilde{X}_N (i,j)$ be a Wigner ensemble and set

\begin{equation}
X_N (i,j) =
\begin{cases}
\tilde{X}_N (i,j) & \text{~  for ~} |i - j| \leq b_N\\
0 & \text{~  otherwise}
\end{cases}
\label{eq:band}
\end{equation}

where $b_N$ is a sequence of integers with $b_N \rightarrow \infty$ and $2 b_N + 1 \leq N$.

We call such matrices banded Wigner matrices with band width $\beta_N = 2 b_N + 1$.

There is a 'Semicircle Law' for banded Wigner matrices due to Bogachev, Molchanov and Pastur \cite{BogachevMP}.
\begin{theorem}
\label{th:BMP}
Suppose $X_N$ is a banded Wigner matrix with band width

$\beta_N = 2 b_N + 1\leq N$ and assume that all moments of $X_N (i,j)$ exist.

Set $M_N = \frac{1}{\sqrt{\beta_N}} X_N$ and denote by $\sigma_N$ the eigenvalue distribution measure of $M_N$.
\begin{enumerate}
\item[1)] If $\beta_N \rightarrow \infty$ but $\frac{\beta_N}{N} \rightarrow 0$ then the $\sigma_N$ converges to the semicircle distribution weakly in probability.
\item[2)] If $\beta_N \approx c N$ for some $c > 0$ then $\sigma_N$ converges weakly in probability to a measure $\tilde{\sigma}$ which is \emph{not} the semicircle distribution.
\end{enumerate}
\end{theorem}

It turns out that the moment method used to prove Wigner's result can also be applied to banded random matrices.

Let us look at the products
\begin{equation*}
X_N (i_1,i_2) \cdot X_N (i_2,i_3) \cdot ~~  \ldots ~~  \cdot X_N (i_k,i_1)
\end{equation*}
which occur in evaluating traces as in \eqref{eq:sumXi}.

We have $N$ possibilities to choose $i_1$. In principle, for $i_2$ we have again $N$ possibilities. However, unlike to the Wigner case, at most $\beta_N$ of these possibilities are not identically zero. This observation makes it plausible that
\begin{equation*}
\sum_{i_{1},\ldots,i_{k}}\; \IE \Big( X_N (i_1,i_2) \cdot ~  \ldots ~  \cdot X_N (i_k,i_1) \Big) ~~\approx~~
 N\, {\beta_N}^{k/2}
\end{equation*}
since - again  - only those terms with each $\{i,j\}$ occurring exactly twice count in the limit. Note that our assumption $\beta_{N}\to\infty$ is needed here.
Without this assumptions pairs $\{i,j\} $ occurring more than twice are not negligible.

Unfortunately, the above argument is not \emph{quite} correct. It \emph{is} true that most columns (and rows) contain $\beta_N$ entries $X_N (i,j)$ which are not identically equal to zero.

However, this is wrong for the rows with row number $j$ when
\begin{equation*}
j \leq b_N \text{~  or ~} j > N - b_N \,,
\end{equation*}
i.e. in the `corners' of the matrix.

Thus for any $1\leq \ell < k$ for which the vertex $i_{\ell+1}$ is new in the path, i.e.~$i_{\ell+1} \notin \{i_1, \ldots, i_\ell\}$ we have ...
$\beta_N$ choices for $i_{l+1}$ \emph{only} if
\begin{equation*}
b_N < i_\ell \leq N - b_N \,.
\end{equation*}
If $\frac{b_N}{N} \rightarrow 0$ (as in case 1 of the theorem) the number of exceptions (i.e. $i_\ell \leq b_N$ or $i_l > N - b_N$) is negligible and the semicircle law is again valid.

However, if $b_N$ grows proportional to $N$ the `exceptional' terms are not exceptional any more but rather contribute in the limit $N \rightarrow \infty$.

For details of the proof see \cite{BogachevMP} or \cite{Catalano}.

The above argument suggests that for $b_N \approx c N$ the limit distribution might be again the semicircle distribution if we `fill the corners' of the matrix appropriately.
This can be achieved by the following modification of \eqref{eq:band}.

\begin{definition}\label{def:band}
Set (for $i \in \IN$)
\begin{equation}\label{eq:NormN}
|\,i\,|_N = \min\, \{ |\,i\,|, |N - i| \}
\end{equation}

and let $\tilde{X}_N$ be a Wigner ensemble. Then we call the matrix

\begin{equation}
X_N (i,j) =
\begin{cases}
\tilde{X}_N (i,j) & \text{~  for ~} |i - j|_N \leq b_N\\
0 & \text{~  otherwise}
\end{cases}
\label{eq:bandper}
\end{equation}

a \emph{periodic band matrix}.
\end{definition}

$|i-j|_N$ measures the distance of $i$ and $j$ on $\mathbb{Z}/N\mathbb{Z}$. The choice of $| \cdot |_N$ guarantees that each column (and each row) contains exactly $\beta_N = 2 b_N + 1$ non zero (i.e. not identically zero) entries.

As we anticipated we have
\begin{theorem}\label{th:BMPperiodic}
If $X_N$ is a periodic band random matrix with band width $\beta_N \leq N$ and
$\beta_N \rightarrow \infty$, then the eigenvalue distribution measure $\sigma_N$ of $\frac{1}{\sqrt{\beta_N}} X_N$ converges weakly in probability to the semicircle distribution $\sigma$.
\end{theorem}

A proof of this result due to Bogachev, Molchanov and Pastur can be found in \cite{BogachevMP} or in \cite{Catalano}.

Catalano \cite{Catalano} has generalized the above result to matrices of the form
\begin{equation}\label{eq:RC1}
X_N (i,j) = \alpha \Big(\frac{|i-j|}{N} \Big) \tilde{X}_N (i,j)
\end{equation}

where $\tilde{X}_N$ is a Wigner matrix and $\alpha : [0,1] \rightarrow \IR$ a Riemann integrable function.

This class of matrices contains both random band matrices with $b_N \approx c N$ and periodic random band matrices, take either $ \alpha (x) = \chi_{[0,c]} (x)$ or $\alpha (x) = \chi_{[0,c] \cup [1-c,1]} (x)$,
where $\chi_{A}(x)=\left\{
                     \begin{array}{ll}
                       1, & \hbox{if $x\in A$;} \\
                       0, & \hbox{otherwise.}
                     \end{array}
                   \right.$

\begin{theorem}\label{th:Catalano}
Let $X_N$ be a matrix ensemble as in \eqref{eq:RC1}, set $$\Phi:=\int_0^1 \int_0^1 \alpha^2 (|x-y|) ~  dx ~  dy $$ and let $\sigma_N$ be the eigenvalue distribution measure for $\frac{1}{\sqrt{\Phi N}} X_N$. Then $\sigma_N$ converges weakly in probability to a limit measure $\tau$.

The limit $\tau$ is the semicircle law if and only if
\begin{align}\label{eq:Catalpha}
   |\alpha(x)|=|\alpha(1-x)|
\end{align}
for almost all $x\in\IR$.
\end{theorem}

Note, that in the case of band matrices with bandwidth proportional to $N$
condition \eqref{eq:Catalpha} is fulfilled for the periodic case, but not for the non periodic case \eqref{eq:band}.

As for the Wigner case the question arises whether the norm of band matrices is bounded in the limit $N\to\infty$. In fact we have:

\begin{theorem}
\label{th:NormBand}
Let $X_N$ be a banded Wigner ensemble (as in \eqref{eq:band}) with band width $\beta_{N}\leq N$ and assume that all moments of $X_{N}(i,j)$ exist.
If there are positive constants $\gamma$ and $C$ such that $\beta_{N}\geq C\,N^{\gamma}$ for all $N$ then
\begin{align}
   \limsup_{N\to\infty} ||\frac{1}{\sqrt{\beta_{N}}}\,X_{N}||~\leq~2
\end{align}
$\IP$-almost surely.
\end{theorem}
A proof of Theorem \ref{th:NormBand} is contained in the forthcoming paper \cite{KK2}. This theorem applies to periodic band matrices as well.

Bogachev, Molchanov and Pastur \cite{BogachevMP} show that the norm of $\frac{1}{\sqrt\beta_{N}}\,X_{N} $ can go to infinity if $\beta_{N} $ grows only
on a logarithmic scale with $N$.

We mention that there are various other results about matrices with independent, but not
identically distributed random variables. Already the papers \cite{Pastur72} and \cite{Pastur73}
consider matrix entries with constant variances but not necessarily identical distribution.
The identical distribution of the entries is replaced by a (far weaker) condition of Lindeberg type.
In the paper \cite{GoetzeNT} even the condition of constant variances is relaxed. Moreover these
authors replace independence by a martingale condition.

\section{Sparse Dependencies}\label{ssec:scl-sparse}
Now we turn to attempts to weaken the assumption of independence between the $X_N (i,j)$ of a matrix ensemble.
We start with what we call `sparse dependencies'.

This means that, while we don't care \emph{how} some of the $X_N (i,j)$ depend on each other, we restrict the \emph{number} of dependencies in a way specified below. We follow Schenker and Schulz-Baldes \cite{SchenkerSB} in this section.

We assume that for each $N$ there is an equivalence relation $\sim_N$ on $\IN_{N}^{2}$ with $\IN_{N}=\{1,2,\ldots,N\}$ and we suppose that the random variables $X_N (i,j)$ and $X_N (k, \ell)$ for $1 \leq j, k \leq \ell$ are independent unless $(i,j)$ and $(k, \ell)$ belong to the same equivalence class with respect to $\sim_N$.
\begin{definition}
\label{def:sparseR}
We call the equivalence relations $\sim_N$ \emph{sparse}, if the following conditions are fulfilled:
\begin{enumerate}
\item[1)] $\underset{i \in \IN_N}{\max}\; |\{(j,k,\ell) \in \IN_N^3 | (i,j) \sim_{N} (k,\ell) \}| =
o(N^2)$
\item[2)] $| \{(i,j,\ell) \in \IN_N^3 | (i,j) \sim_{N} (j,\ell)$ and $\ell \neq i \} | = o(N^2)$
\item[3)] $ \underset{i,j,k \in \IN_N}{\max} | \{\ell \in \IN_N | (i,j) \sim_{N} (k, \ell) \} | \leq B$
\quad for an $N$-independent constant $B$.
\end{enumerate}
\end{definition}
\begin{definition}
A symmetric random matrix ensemble $X_N (i,j)$ with

$\IE \Big(X_N (i,j) \Big) = 0, \IE \Big(X_N (i,j)^2 \Big) = 1$ and $\underset{N,i,j}{\sup} ~  \IE \Big(X_N (i,j)^k \Big) < \infty$ for all $k \in \IN$ is called a \emph{generalized Wigner ensemble with sparse dependence structure} if there are sparse equivalence relations $\sim_N$, such that $X_N (i,j)$ and $X_N (k, \ell)$
are independent if $(i,j) \nsim_N (k, \ell)$.
\end{definition}
\begin{examples}
If $A_N$ and $B_N$ are Wigner matrices, then the $2N \times 2N$-matrices

\begin{align*}
X_N =
\begin{pmatrix}
A_N & B_N\\
B_N & -A_N
\end{pmatrix}
\end{align*}
and
\begin{align*}
X_N^{'} =
\begin{pmatrix}
A_N & B_N\\
B_N & A_N
\end{pmatrix}
\end{align*}

are generalized Wigner ensembles with sparse dependence structure.

Many more example classes can be found in \cite{HofmannS}.
\end{examples}

\begin{theorem}
\label{th:SHSB}
If $X_N$ is a generalized Wigner ensemble with sparse dependence structure and $\sigma_N$ is the eigenvalue distribution measure of $M_N = \frac{1}{\sqrt{N}} X_N$ then $\sigma_N$ converges to the semicircle distribution weakly in probability.
\end{theorem}
This theorem is due to Schenker and Schulz-Baldes \cite{SchenkerSB} who proved weak convergence in expectation, for convergence in probability see \cite{Catalano}.

Catalano combines sparse dependence structures with generalized band structures as in \eqref{eq:RC1}. In fact, he proves Theorem \ref{th:Catalano} also for such matrix ensembles. For details we refer to \cite{Catalano}.

\section{Decaying Correlations}\label{ssec:scl-deccor}
In this section we discuss some matrix ensembles for which the random variables $X_{N}(i,j)$
have decaying correlations. We begin by what we call `diagonal' ensembles. By this we mean that
the random variables $X_{N}(i,j)$ and $X_{N}(i',j')$ are independent if the index pairs $(i,j)$ and
$(i',j')$ belong to different diagonals, i.\,e. if $i-j\not=i'-j'$ (for $i\leq j$ and $i'\leq j'$).

\begin{definition}
   Suppose $Y_{n}$ is a sequence of random variables and ${Y_{n}}^{(\ell)}, \ell\in\IN$ are independent
   copies of $Y_{n}$, then the matrix ensemble
   \begin{align}
      X_{N}(i,j)~=~Y_{i}^{(|i-j|)}\qquad\text{for } 1\leq i\leq j \leq N
   \end{align}
   is called the \emph{matrix ensemble with independent diagonals} generated by $Y_{n}$.
\end{definition}

Of course, if the random variables $Y_{n}$ themselves are independent then we obtain an
independent matrix ensemble. If, on the other hand, $Y_{n}=Y_{1}$ we get a matrix with
constant entries along each diagonal, which vary randomly from diagonal to diagonal.
Such a matrix is thus a \emph{random Toeplitz matrix}. Random Toeplitz matrices were
considered by Bryc, Dembo and Jiang in \cite{BrycDJ}. They prove:

\begin{theorem}\label{th:BDJ}
  Suppose that $X_{N}(i,j)$ is the random Toeplitz matrix ensemble associated with $Y_{n}=Y$
  with $\IE(Y)=0$, $\IE(Y^{2})=1$ and $\IE(Y^{K})<\infty$ for all $K$, then the eigenvalue distribution measures $ \sigma_{N}$
  of $\frac{1}{\sqrt{N}}\,X_{N}$ converge weakly almost surely to a nonrandom measure $\gamma$ which is independent of the
  distribution of\, $Y$ and has unbounded support.

  In particular, $\gamma$ is not the
  semicircle distribution.
\end{theorem}

Friesen and L\"{o}we \cite{FriesenLoewe1} consider matrix ensembles with independent diagonals generated by a
sequence $Y_{n}$ of weakly correlated random variables. In their case the limit distribution
is the semicircle law again.

\begin{theorem}\label{th:FL1}
   Let $Y_{n}$ be a stationary sequence of random variables with $\IE(Y_{1})=0$, $\IE({Y_{1}}^{2})=1$
   and $\IE({Y_{1}}^{K}) <\infty$ for all $K$. Assume
   \begin{align}
        \sum_{\ell=1}^{\infty}\,\big|\IE(Y_{1}Y_{1+\ell})\big|~<~\infty\,.
   \end{align}
   Let $X_{N}$ be the matrix ensemble with independent diagonals generated by $Y_{n}$.
   Then the eigenvalue distribution measures $ \sigma_{N}$
  of $\frac{1}{\sqrt{N}}\,X_{N}$ converge to the semicircle distribution $\IP$-almost surely.
\end{theorem}

The next step away from independence is to start with a sequence $\{Z_{n}\}_{n\in\IN}$ of random variables
and to distribute them in some prescribed way on the matrix entries $X_{N}(i,j)$. It turns out (see \cite{LoeweS}) that the validness of the semicircle law depends on the way we
fill the matrix with the random number $Z_{n}$.

One main example of a filling is the `diagonal' one, resulting in:
\begin{align}\label{eq:diagonal}
   X_{N}~=~
   \begin{pmatrix}
 Z_{1} & Z_{N+1} & Z_{2N}&\cdots&\cdots&\cdots & Z_{\frac{N(N+1)}{2}}\\
 Z_{N+1} & Z_{2} & Z_{N+2}&\cdots & \cdots&\cdots&\cdots\\
  Z_{2N} & Z_{N+2} & Z_{3}& Z_{N+3} & \cdots&\cdots&\cdots\\
  \cdots & Z_{2N+1} & Z_{N+3}&Z_{4}& \cdots&\cdots&\cdots\\
 \vdots    & \vdots    &       & & \vdots & &  \\
 Z_{\frac{N(N+1)}{2}}& \cdots & \cdots&\cdots &\cdots& Z_{2N-1}& Z_{N}
\end{pmatrix}
\end{align}

L\"{o}we nd Schubert define abstractly:
\begin{definition}
   A \emph{filling} is a sequence of bijective mappings
   \begin{align}
      \varphi_{N}: \{1,2,\ldots,\frac{N(N+1)}{2}\}~\longrightarrow~\{(i,j)\in\{1,2,\ldots,N\}^{2}\mid i\leq j\}
   \end{align}
   If $Z_{n}$ is a stochastic process and $\{\varphi_{N}\}$ is a filling we say that
   \begin{align}
      X_{N}(i,j)~=~Z_{{\varphi_{N}}^{-1}(i,j)}\qquad \text{for } 1\leq i\leq j\leq N
   \end{align}
   is the matrix ensemble corresponding to $\{Z_{n}\}$ with filling $\{\varphi_{N}\}$
\end{definition}
Another example of a filling, besides the `diagonal' one, is the (symmetric) `row by row' filling:
\begin{align}
   X_{N}~=~
   \begin{pmatrix}
 Z_{1} & Z_{2} & Z_{3}&\cdots&\cdots&\cdots & Z_{N}\\
 Z_{2} & Z_{N+1} & Z_{N+2}&Z_{N+3} & \cdots&\cdots&Z_{2N-1}\\
  Z_{3} & Z_{N+2} & Z_{2N}& Z_{2N+1} & \cdots&\cdots&Z_{3N-3}\\
  \cdots & Z_{N+3} & Z_{2N+1}&Z_{3N-1}& \cdots&\cdots&\cdots\\
 \vdots    & \vdots    &       & & \vdots & &  \\
 Z_{N}& Z_{2N-1} & \cdots&\cdots &\cdots& \cdots& Z_{\frac{N(N+1)}{2}}
\end{pmatrix}
\end{align}

Among other results, L\"{o}we and Schubert prove:
\begin{theorem}\label{th:LS}
   Suppose $Z_{n}$ is an ergodic Markov chain with finite state space
$S\subset\IR$ started in its stationary measure and assume
\begin{align}
   &\IE(Z_{n_{1}}Z_{n_{2}}\ldots Z_{n_{k}})~=~0\label{eq:LS1}\\
   &\IE({Z_{n}}^{2})~=~1\label{eq:LS2}
\end{align}
for any $n$ and any $n_{1},\ldots,n_{k}$ with $k$ odd.

If $X_{N}$ is the matrix ensemble corresponding to $\{Z_{n}\}$ with diagonal filling
then the eigenvalue distribution measures $ \sigma_{N}$
  of $\frac{1}{\sqrt{N}}\,X_{N}$ converge to the semicircle distribution $\IP$-almost surely.
\end{theorem}

The assumptions we made in Theorem \ref{th:LS} both on $Z_{n}$ and on the filling are but an
example of the abstract assumptions given in \cite{LoeweS}. These authors also show:

\begin{theorem}
   There is an ergodic Markov chain $\{Z_{n}\}$ with finite state space
$S\subset\IR$ started in its stationary measure satisfying \eqref{eq:LS1} and \eqref{eq:LS2}
such that for the matrix ensemble $X_{N}$ corresponding to $\{Z_{n}\}$ with row by row filling
the eigenvalue distribution measures $ \sigma_{N}$
  of $\frac{1}{\sqrt{N}}\,X_{N}$ do \emph{not} converge to the semicircle distribution.
\end{theorem}

Consequently, the convergence behavior of $\sigma_{N}$ depends not only on the process $\{Z_{n}\}$
but also on the way we fill the matrices with this process. For details we refer to \cite{LoeweS}.

\section{Curie-Weiss Ensembles}\label{ssec:scl-CW}
In section \ref{ssec:scl-deccor} we discussed matrix ensembles $X_N (i,j)$ which are generated through stochastic processes with decaying correlations. Thus, for fixed $N$, the correlations
$\IE\,\Big(X_N (i,j) X_N (k, \ell) \Big)$ become small for $(i,j)$ and $(k, \ell)$ far apart, in some appropriate sense.

In the present section we investigate matrix ensembles $X_N (i,j)$ with $\IE\big(X_{N}(i,j)\big)=0$ for which the correlations $\IE\,\Big(X_N (i,j) X_N (k, \ell) \Big)$ do \emph{not} depend on $i, j, k, \ell$ for most (or  at least many) choices of $i, j, k$ and $\ell$, but the correlations depend on $N$ instead.

More precisely, we will have that for given $(i,j)$
\begin{equation*}\label{Teq:corr_N}
\IE ~  \Big(X_N (i,j) X_N (k, \ell) \Big) \sim C_N \geq 0
\end{equation*}
for $(k, \ell) \in B_N$ with $|B_N| \sim N$ or even $|B_N| \sim N^2$, and, as a rule,  $ C_N \rightarrow 0$. However, in Theorem \ref{th:KK} we will encounter an example
for which $C_{N}$ does not decay.

The main example we discuss comes from statistical physics, more precisely from the Curie-Weiss model.

\begin{definition}\label{def:CWM}
Curie-Weiss random variables $\xi_1, \ldots, \xi_M$ take values in $\{-1,1\}^M$ with probability
\begin{equation}
\IP_{\beta}^M (\xi_1 = x_1, \ldots, \xi_M = x_M) = Z^{-1} e^{\frac{\beta}{2M} (\sum_{i=1}^M x_i)^2}
\end{equation}

where $Z=Z_{\beta,M}$ is a normalization constant  (to make $\IP_{\beta}^M$ a  probability measure) and $\beta \geq 0$ is a parameter which is interpreted in physics as `inverse temperature', $\beta = \frac{1}{T}$.
\end{definition}
If $\beta = 0$ ~  ($T = \infty)$ the random variables $\xi_i$ are independent while for $\beta > 0$ there is a positive correlation between the $\xi_i$, so the $\xi_i$ tend to have the same value $+1$ or $-1$. This tendency is growing as $\beta \rightarrow \infty$. The Curie-Weiss model is used in physics as an easy model to describe magnetism. The $\xi_i$ represent small magnets (`spins') which can be directed upwards (`$\xi_i = 1$') or downwards (`$\xi_i = -1$'). At low temperature (high $\beta$) such systems tend to be aligned, i.\,e. a majority of the spins have the same direction (either upwards or downwards). For high temperature they behave almost like independent spins . These different types of  behavior are described in the following theorem.

\begin{theorem}
\label{th:CWLLN}
Suppose $\xi_1, \ldots , \xi_M$ are $\IP_{\beta}^M$-distributed Curie-Weiss random variables. Then the mean $\frac{1}{M} \sum_{i=1}^M \xi_i$ converges in distribution, namely

\begin{equation}\label{eq:CWLLN}
\frac{1}{M}\, \sum_{i=1}^M\; \xi_i ~~\overset{\mathcal{D}}{\Longrightarrow}~~
\begin{cases}
~\delta_0 & \text{~ if ~} \beta \leq 1\\
~\frac{1}{2} (\delta_{-m (\beta)} + \delta_{m (\beta)}) & \text{~  if ~} \beta > 1
\end{cases}
\end{equation}

where $m = m (\beta)$ is the (unique) strictly positive solution of
\begin{equation}\label{eq:defm}
\tanh (\beta m) = m
\end{equation}
\end{theorem}
Above we used $\overset{\mathcal{D}}{\Longrightarrow}$ to indicate convergence in distribution: Random variables $\zeta_i$ \emph{converge in distribution} to a measure $\mu$ if the distributions of $\zeta_i$ converge weakly to $\mu$. Also, $\delta_x$ denotes the Dirac measure (see \eqref{eq:delta}).

For a proof of the above theorem see e.\,g. \cite{Ellis} or \cite{MoBu}.

Theorem \ref{th:CWLLN} makes the intuition from physics precise: The $\xi_i$ satisfy a law of large numbers, like independent random variables do, if $\beta \leq 1$, in the sense that the distribution  of $m_M = \frac{1}{M} \sum_{i=1}^M \xi_i$ converges weakly to zero, while $m_M$, the `mean magnetization', equals $\pm m (\beta) \neq 0$ in the limit, with probability $\frac{1}{2}$ each, for $\beta > 1$.

In physics jargon, there is a phase transition for the Curie-Weiss model at $\beta=1$, the `critical inverse temperature'.

We now discuss two matrix ensembles connected with Curie-Weiss random variables. The first one,
which we call the \emph{diagonal Curie-Weiss ensemble}, was introduced in \cite{FriesenLoewe2}. It has
independent `diagonals' and the matrix entries within the same diagonal are Curie-Weiss distributed.
Thus, it is closely related to the diagonal filling as defined in \eqref{eq:diagonal}.

\begin{definition}\label{def:diagCW}
  Let the random variables $\xi_{1},\xi_{2},\ldots,\xi_{N}$ be $\IP_{\beta}^{N}$-distributed Curie-Weiss random variables and take $N$ independent copies of the $\xi_{i}$, which we call

${\xi_{1}}^{1},{\xi_{2}}^{1},\ldots,{\xi_{N}}^{1}$,\quad
   ${\xi_{1}}^{2},{\xi_{2}}^{2},\ldots,{\xi_{N}}^{2}$ ,\quad\ldots, \quad
   ${\xi_{1}}^{N},{\xi_{2}}^{N},\ldots,{
\xi_{N}}^{N}$.

Then we call the random matrix
\begin{align}
   X_{N}(i,i+\ell)~&:=~{\xi_{i}}^{\ell}\quad&&\text{for $\ell=0,\ldots,N-1$ and $i=1,\ldots,N-\ell$}\\
   X_{N}(i,j)~&:=~X_{N}(j,i)\quad&&\text{for $i>j$}
\end{align}
the \emph{diagonal Curie-Weiss ensemble} (with diagonal distribution $\IP_{\beta}^{N}$).
\end{definition}

For the diagonal Curie-Weiss ensemble Friesen and L\"{o}we \cite{FriesenLoewe2} prove the following result.

\begin{theorem}\label{th:CWdiag}
  Suppose $X_{N}$ is a diagonal Curie-Weiss ensemble with diagonal distribution $\IP_{\beta}^{N}$.

Then the eigenvalue distribution measure $\sigma_{N}$ of $\frac{1}{\sqrt{N}}X_{N}$ converges weakly
almost surely to a measure $\sigma_{\beta}$.
$\sigma_{\beta}$ is the semicircle law $\sigma$ if and only if $\beta\leq 1$.
\end{theorem}
\begin{remarks}
   \begin{enumerate}
      \item The theorem shows that there is a phase transition for the eigenvalue distribution of the diagonal Curie-Weiss ensemble at $\beta=1$.

      \item The proof in \cite{FriesenLoewe2} uses the moment method. It allows the authors to give
        an expression for the moments of $\sigma_{\beta}$ in terms of $m(\beta)$ (see \eqref{eq:CWLLN}).
   \end{enumerate}
\end{remarks}
For large $\beta$ the eigenvalue distribution measure of the diagonal Curie-Weiss ensemble approaches the eigenvalue distribution measure of random Toeplitz matrices we discussed in Theorem \ref{th:BDJ} (see Bryc, Dembo and Jiang \cite{BrycDJ}).

The second Curie-Weiss-type matrix ensemble, which we call the `full Curie-Weiss ensemble', is defined as follows.
\begin{definition}
Take $N^2$ Curie-Weiss random variables $\tilde{X}_N(i,j)$ with distribution $\IP_{\beta}^{N^2}$ and set
\begin{align}
X_N (i,j) =
\begin{cases}
\tilde{X}_N (i,j) & \text{~ for ~} i \leq j\\
\tilde{X}_N (j,i) & \text{~  otherwise} \,.
\end{cases}
\end{align}
We call the random matrix $X_{N}$ defined above the \emph{full Curie-Weiss ensemble}.
\end{definition}

To our knowledge this ensemble was first considered in \cite{HKW}, where the following
result was proved.
\begin{theorem}
\label{th:HKWKK}
Let $X_N$ be the full Curie-Weiss matrix ensemble with inverse temperature $\beta\le 1$.
Then the eigenvalue distribution measure $\sigma_N$ of $\frac{1}{\sqrt{N}} X_N$ converges weakly in probability to the semicircle distribution $\sigma$.
\end{theorem}

The proof is based on the moment method we discussed in section \ref{ssec:scl-wigner}.
In \cite{HKW} the authors prove this result just using assumptions on correlations of the
$X_{N}(i,j)$ which are in particular satisfied by the full Curie-Weiss model if $\beta\leq 1$.
Here, we only discuss this special case and refer to \cite{HKW} for the more general case.

The main difficulty in this proof is the fact that for the Curie-Weiss ensemble it is \emph{not} true that
\begin{equation}\label{eq:prod}
\IE \Big( X_N (i_1,i_2) \cdot X_N (i_2,i_3) \cdot ~  \ldots ~  \cdot  X_N (i_k,i_{1}) \Big)
\end{equation}

is zero if an edge $\{i, j\}$ occurs only once in \eqref{eq:prod} (cf. \eqref{eq:EXi} for the independent case). In other words, we need an appropriate substitute for Lemma \ref{lem:once}.

So, we need a way to handle expectations as in \eqref{eq:prod} when there are edges (=index pairs,
see Definition \ref{def:graph}) which occur only once. Let us call such index pairs `single edges'.

Correlation estimates as we need them can be obtained from a special way of writing expectations $\IE_{\beta}^M$ with respect to the measure $\IP_{\beta}^M$.

\begin{definition}
\label{def:Pt}
For $t \in [-1,1]$ we denote by $P_t^{(1)}$ the probability measure on $\{-1,1\}$ given by
\begin{equation*}
P_t^{(1)} (1) = \frac{1}{2} (1+t) \text{~  and ~} P_t^{(1)} (-1) = \frac{1}{2} (1-t)\,.
\end{equation*}
$P_t^{(M)}$ denotes the $M$-fold product of $P_t^{(1)}$ on $\{-1,1\}^M$. If $M$ is clear from the context we write $P_t$ instead of $P_t^{(M)}$.

By $E_t$ resp. $E_t^{(M)}$ we denote the corresponding expectation.
\end{definition}
\begin{proposition}
\label{prop:dFCW}
For any function $\phi$ on $\{-1,1\}^M$ we have
\begin{equation}\label{eq:dFCW}
\IE_{\beta}^M \Big(\phi (X_1, \ldots, X_M) \Big) = \int_{-1}^1 E_t \Big( \phi (X_1, \ldots , X_M) \Big) \frac{e^{-M \frac{1}{2} F_{\beta} (t)}}{1-t^2} ~  dt
\end{equation}

where $F_{\beta} (t) = \frac{1}{\beta} (\frac{1}{2} \ln \frac{1+t}{1-t})^2 + \ln (1-t^2)$.

\end{proposition}
This proposition can be proved using the so called Hubbard-Stratonovich transformation. For a proof see \cite{HKW} or \cite{MoBu}. The way to write expectations with respect to $\IP_{\beta}^{M}$ as a
combination of independent measure is typical for exchangeable random variables and is known as
de Finetti representation \cite{deFinetti}. We will discuss this issue in detail in Section \ref{ssec:scl-ee} and in particular in \cite{KK2}.

The advantage of the representation \eqref{eq:dFCW} comes from the observation that under the probability measure $P_t$ the random variables $X_1, \ldots, X_M$ are \emph{independent} and the fact that the integral is in a form which is immediately accessible to the Laplace method for the asymptotic evaluation of integrals.

The Laplace method and Proposition \ref{prop:dFCW} yield the required correlation estimates.

\begin{proposition}
\label{prop:Corr}
Suppose $X_1, \ldots, X_M$ are $\IP_{\beta}^M$-distributed Curie-Weiss random variables.

If $\ell$ is even, then as $M\to\infty$
\begin{enumerate}
\item if $\beta < 1$
\begin{equation*}
\IE_{\beta}^{(M)} (X_1 \cdot X_2\cdot \ldots \cdot X_{\ell})~ \approx~ (l-1)!!\ \Big(\frac{\beta}{1- \beta} \Big)^{\frac{\ell}{2}} \frac{1}{M^{\frac{\ell}{2}}}
\end{equation*}
\item if $\beta = 1$ there is a constant $c_{\ell}$ such that
\begin{equation*}
\IE_{\beta}^{(M)} (X_1 \cdot X_2\cdot \ldots \cdot X_{\ell}) ~\approx~ c_{\ell}\ \frac{1}{M^{\frac{\ell}4}}
\end{equation*}
\item \label{Corrg1} if $\beta > 1$
\begin{equation*}
\IE_{\beta}^{(M)} (X_1 \cdot X_2\cdot \ldots \cdot X_{\ell}) ~\approx~ m (\beta)^{\ell}
\end{equation*}
where $t= m (\beta)$, as in \eqref{eq:defm}, is the strictly positive solution of $\tanh \beta t = t$.
\end{enumerate}
If $\ell$ is odd then $\IE_{\beta}^{(M)} (X_1 \cdot X_2\cdot\ldots \cdot X_{\ell}) = 0$ for all $\beta$.
\end{proposition}
We remind the reader that for an odd number $k$ we set $k!!=k\cdot (k-2)\cdot\ldots\cdot 3 \cdot 1$.

For proof of Proposition \ref{prop:Corr} see again \cite{HKW} or \cite{MoBu}.

From Proposition \ref{prop:Corr} we get immediately the following Corollary, which substitutes
Lemma \ref{lem:once}.

\begin{corollary}\label{cor:cwcorr}
   Let $X_{N}$ be the full Curie-Weiss matrix ensemble with inverse temperature $\beta$
and let the graph corresponding to the sequence $i_{1},i_{2},\ldots,i_{k}$ contain $\ell$
single edges.
\begin{enumerate}
   \item If $\beta<1$ then
\begin{equation}\label{eq:EXi1}
\Big|\,\IE \Big(X_N (i_1,i_2) \cdot X_N (i_2,i_3) \cdot \ldots  \cdot X_N (i_k,i_1) \Big)\,\Big|~\leq~C\,N^{-\ell}\,.
\end{equation}
\item If $\beta=1$ then
\begin{equation}\label{eq:EXi2}
\Big|\,\IE \Big(X_N (i_1,i_2) \cdot X_N (i_2,i_3) \cdot \ldots  \cdot X_N (i_k,i_1) \Big)\,\Big|~\leq~C\,N^{-\ell/2}\,.
\end{equation}
\end{enumerate}
\end{corollary}

In the next step we have to prove a quantitative version of Proposition \ref{prop:r-gross}.

\begin{proposition}\label{prop:r-gross-genau}
    If $|\{i_1, \ldots , i_k \}|\geq 1+k/2 +s $ for some $s>0$, then there are at least $2s+2$ single edges in
$\{i_{1},i_{2}\},\{i_{2},i_{3}\},\ldots,\{i_{k},i_{1}\} $.
\end{proposition}

\begin{proof}
   The proof is a refinement of the proof of Proposition \ref{prop:r-gross}.

Suppose $\mathcal{G}$ is a multigraph with $r$ vertices and $k$ edges.
Then, as we saw already, $k\geq r-1$, if $\mathcal{G}$ is connected. So, there are at most $k-r+1$ edges
left for `double' connections. This means that there are at least $\ell=r-1-(k-r+1)$ single
edges and
\begin{align}
   \ell~=~r-1-(k-r+1)~&=~ 2r-k-2\notag\\
                &~\geq (k+2+2s)-k-2\notag\\
                &=~ 2s\label{eq:rk}
\end{align}
by assumption on $r$. So, by the above simple argument we are off the assertion by two only.

Now, we take into account that the sequence $(i_1, \ldots , i_k, i_{1})$ defines a closed
path through the graph. Since $|\{i_1, \ldots, i_k\}| > 1 + k/2$ there is at least one single edge.
If we remove one of the single edges from the graph, this new graph $\mathcal{G}'$
is still connected. $\mathcal{G}'$ has $r$ vertices and $k-1$ edges.

We redo the above argument with  the graph $\mathcal{G}'$ and get for the
minimal number $\ell'$ of single edges in $\mathcal{G}'$ equation \eqref{eq:rk} with $k$ replaced
by $k-1$ and thus obtain
\begin{align}
   \ell'~=~r-1-(k-1-r+1)~&=~ 2r-k-1\notag\\
                &\geq~ k+2+2s-k-1\notag\\
                &=~2s+1
\end{align}
Since we have removed a single edge from $\mathcal{G}$, the graph $\mathcal{G}$ has at least $2s+2$ single edges.
\end{proof}
Corollary \ref{cor:cwcorr} and Proposition \ref{prop:r-gross-genau} together allow us to do the
moment argument as in Section \ref{def:Wigner}.

We turn to the case $\beta>1$ for the full Curie-Weiss model. Part \ref{Corrg1} of Proposition \ref{prop:Corr} shows that there are strong correlations in this case, so one
is tempted to believe that there is no semicircle law for $\beta>1$.

In fact, it is easy to see that for $\beta>1$ the expectations of
$\frac{1}{N^{1+k}}\,\text{tr}({X_{N}}^{2k})$ cannot
converge for $k\geq 2$ as $N\to\infty$. For example, for $k=2$ we have
\begin{align}
   &\qquad\IE\,\Big(\frac{1}{N^{3}}\,\text{tr}\big({X_{N}}^{4}\big)\Big)~\notag\\
   &=~
   \frac{1}{N^{3}}\,\sum_{\ueber[1]{i_{1},i_{2},i_{3},i_{4}}{\text{all different}}}
   \IE\,\Big(X_{N}(i_{1},i_{2})X_{N}(i_{2},i_{3})X_{N}(i_{3},i_{4})X_{N}(i_{4},i_{1})\Big)+\mathcal{O}(1)\notag\\[3mm]
   &\approx~\frac{N(N-1)(N-2)(N-3)}{N^{3}}\;m(\beta)^{4}~\rightarrow~\infty \label{eq:momgro}
\end{align}
so the moment method will not work here.

A closer analysis of the problem shows that the divergence of the moments of traces is due to a \emph{single}
eigenvalue of $\frac{1}{\sqrt{N}}X_{N}$ which goes to infinity. All the other eigenvalues
behave `nicely'.
Informally speaking, for $\beta>1$ the matrices $X_{N}$ fluctuate around the matrices
$\pm m(\beta)\,\mathcal{E}_{N}$ (see \eqref{eq:defE}) with probability $1/2$ each. As we saw
in Section \ref{ssec:scl-setup} these matrices have rank one. So, one may hope that they do not
change the eigenvalue distribution measure in the limit.

Analyzing the fluctuations around $\pm m(\beta)\,  \mathcal{E}_{N}$ one can apply the moment
method to $X_{N}\mp m(\beta)\,\mathcal{E}_{N}$. The variance of the matrix entries is $v(\beta)=1-m(\beta)^{2}$,
so this has a chance to converge to the semicircle distribution, but scaled due to the variance $v(\beta)<1$. In fact we have:

\begin{theorem}
\label{th:KK}
Let $X_N$ be the full Curie-Weiss matrix ensemble with arbitrary inverse temperature $\beta\geq 0$.
Then the eigenvalue distribution measure $\sigma_N$ of $\frac{1}{\sqrt{N}} X_N$ converges weakly in probability to the rescaled semicircle distribution $\sigma_{v(\beta)}$, given by:

  \begin{equation}\label{eq:defscd1}
\sigma_{v(\beta)} (x) =
\begin{cases}
\frac{1}{2 \pi v(\beta)} \sqrt{4v(\beta) - x^2} & \text{~  , for ~} |x| \leq 2\sqrt{v(\beta)}\\
0 & \text{~  , otherwise.}
\end{cases}
\end{equation}
 Here, $v(\beta)=1-m(\beta)^2$ with $m(\beta)=0$ for $\beta\leq 1$ and $m=m(\beta)$ is the unique positive solution of
$\tanh(\beta m)=m$ for $\beta > 1$ (cf.~\eqref{eq:defm}).
\end{theorem}

A detailed proof will be contained in \cite{KK1}.

Already in \eqref{eq:momgro} we saw that the norm of $\frac{1}{\sqrt{N}} X_N$ does not converge for the full Curie-Weiss ensemble if $\beta>1$. This is
made precise in the following theorem.

\begin{theorem}\label{th:CWNorm}
   Suppose $X_{N}$ is a full Curie-Weiss ensemble.
   \begin{enumerate}
      \item \label{th:CWNi} If $\beta<1$ then
      \begin{equation*}
\| \frac{1}{\sqrt{N}}\,X_{N} \| \rightarrow 2 \text{~  as ~} N \rightarrow \infty
\end{equation*}
$\IP$-almost surely.
\item \label{th:CWNii} If $\beta=1$ then
      \begin{equation*}
\| \frac{1}{N^{\gamma}}\,X_{N} \| \rightarrow 0 \text{~  as ~} N \rightarrow \infty
\end{equation*}
for every $\gamma>1/2\quad\IP$-almost surely.
 \item \label{th:CWNiii} If $\beta>1$ then
      \begin{equation*}
\| \frac{1}{N}\,X_{N} \| \rightarrow m(\beta) \text{~  as ~} N \rightarrow \infty
\end{equation*}
$\IP$-almost surely.
   \end{enumerate}
\end{theorem}

Theorem \ref{th:CWNorm}.\ref{th:CWNiii} was proved in \cite{HKW}, \ref{th:CWNi} and \ref{th:CWNii} can be found in \cite{KK1}.

\section{Ensembles with Exchangeable Entries}\label{ssec:scl-ee}

The results presented in the previous section for Curie-Weiss ensembles with subcritical temperatures ($\beta >1$) suggest that models
with correlations that do not decay sufficiently fast as $N$ tends to infinity (e.g. in the sense of Corollary \ref{cor:cwcorr})
may display a wealth of spectral phenomena depending on the specific features of the model. This is largely uncharted territory.
One step into this world is to consider matrix ensembles with entries chosen from a sequence of exchangeable random variables.

A sequence $(\xi_i)_{i\in \IN}$ of real valued random variables with underlying probability space $(\Omega, \mathcal{F}, \IP)$ is called {\em exchangeable},
if for all integers $N \in \mathbb N$, all permutations $\pi$ on $\{1, \ldots, N\}$, and $F \in \mathcal{B}(\IR^{N})$ it is true that
\begin{equation*}
\IP \big( (\xi_1, \ldots, \xi_N) \in F \big) =  \IP \big( (\xi_{\pi(1)}, \ldots, \xi_{\pi(N)}) \in F \big) \,.
\end{equation*}
Generalizing a result of de Finetti \cite{deFinetti, deFinetti2} for random variables that only take on two values,
Hewitt and Savage \cite[Theorem 7.4]{HewittSavage}
showed in a very general setting that such probability measures $\IP$ may be represented as averages of i.i.d.~sequences
with respect to some probability measure $\mu$. In our context we impose the additional condition that
all moments of the random variables $\xi_i$ exist (cf. Definition \ref{def:moments}).
This leads us to the following general definition of ensembles of real symmetric matrices with exchangeable entries.
\begin{definition}\label{def:eee}
Let $\mu$ denote a probability measure on some measurable space $(T, \mathcal{T})$
and let $\Lambda : T \to \mathcal{M}_1^{(0)} (\IR)$ be a measurable map that assigns every element
$\tau$ of $T$ to a  Borel probability measure $\Lambda_{\tau}$ on $\IR$
for which all moments exist (we call $\mathcal{M}_1^{(0)} (\IR)$
the set of all such probability measures on $\IR$). Define
\begin{equation}\label{eq:Tex.1}
\IP^{\,\mu,\Lambda} := \int_T  P_{\tau} d \mu (\tau) \,, \quad \mbox{with} \quad
P_{\tau} := \bigotimes_{i=1}^{\infty} \Lambda_{\tau}\,,
\end{equation}
as the $\mu$-average of i.i.d.~sequences of real random variables with distributions $\Lambda_{\tau}$.
The corresponding {\em matrix ensemble with exchangeable entries}
is given by matrices $X_N$ with entries $X_N(i,j)$, $1 \leq i \leq j \leq N$,
that equal the first $N(N+1)/2$ members of the sequence $(\xi_i)_i$ of exchangeable random variables. The remaining entries
$X_N(i,j)$, $1 \leq j < i \leq N$ are then fixed by symmetry $X_N(i,j)=X_N(j,i)$. Observe that due to the exchangeability of $(\xi_i)_i$
it is of no relevance in which order the upper triangular part of $X_N$ is filled by $\xi_1, \ldots, \xi_{N(N+1)/2}$.
Moreover, one could have chosen any $N(N+1)/2$ distinct members of $(\xi_i)_i$ to fill the entries of $X_N$
without changing the ensemble.
\end{definition}

It is instructive to consider the special case of ensembles that allow only for matrix entries $X_N(i,j) \in \{1, -1\}$.
We refer to it as the \emph{spin case}. Observe that the probability measures with support contained in $\{1, -1\}$
are all represented by the family
$\Lambda_{\tau} = \frac{1}{2} \left[ (1+\tau) \delta_1 + (1-\tau)\delta_{-1}\right]$, $\tau \in T := [-1,1]$.
Hence all ensembles of the spin case are given by \eqref{eq:Tex.1} with the just mentioned choices for $T$ and $\Lambda_{\tau}$.
They are parameterized by the probability measures $\mu$ on $[-1,1]$.
Recall that $\Lambda_{\tau}$ already appeared in
Definition \ref{def:Pt} as the building block $P_t^{(1)}$ for Curie-Weiss ensembles. What is different from
Section \ref{ssec:scl-CW} is that there the averaging measure $\mu$ depends on the matrix size $N$ and is of a special form.

Let us return to the general ensembles with exchangeable entries of Definition \ref{def:eee}.
The key for analyzing both the eigenvalue distribution measure and the operator norm is
that for every $\tau \in T$ the measure $P_{\tau}$ generates i.i.d.~entries for $X_N$. For the latter ensembles $P_{\tau}$
the following observations that can already be found in \cite{FuerediK} are useful:
Subtracting the mean of the entries yields a Wigner ensemble (multiplied by the standard
deviation of $\Lambda_{\tau}$) for which Theorem \ref{th:Wigner} is applicable.
Considering first the eigenvalue distribution measure we note that the mean is some
multiple of the matrix $\mathcal{E}_N$ defined in \eqref{eq:defE}. Since $\mathcal{E}_N$ has rank 1 the subtraction of the mean
will not have an influence on the limiting spectral measure.
As $\IP^{\,\mu,\Lambda}$ is the $\mu$-average over all measures $P_{\tau}$ it is plausible that the limit of the eigenvalue
distribution measures is an average of scaled semicircles w.r.t. the measure $\mu$, where the scaling factors are given by
the standard deviaton of $\Lambda_{\tau}$. Accordingly, we define
\begin{equation}
\label{eq:B2.1}
\sigma_{\mu} := \int_{T} \sigma_{v(\tau)} ~  d\mu (\tau)\,,
\end{equation}
where $v(\tau)$ denotes the variance of $\Lambda_{\tau}$ and $\sigma_v$ is the semicircle distribution with support
$[-2 \sqrt{v}, 2 \sqrt{v}]$ (cf. Definition \eqref{eq:defscd1}). We prove in \cite{KK2}
\begin{theorem}
\label{theo:B3}
Denote by $\IP^{\, \mu,\Lambda}, \sigma_{\mu}$ the measures introduced in Definition \ref{def:eee}
and in \eqref{eq:B2.1}. Then the eigenvalue distribution measures $\sigma_N$ of $X_N / \sqrt{N}$
converge weakly in expectation to $\sigma_{\mu}$ w.r.t. to the measure $\IP^{\, \mu,\Lambda}$.
\end{theorem}
Moreover, it is shown in \cite{KK2} that $\sigma_{\mu}$ is a semi-circle if and only if
the function $\tau \mapsto v(\tau)$ is constant $\mu$-almost surely.

For the operator norm the situation is quite different.
Since $\| \mathcal{E}_N \| = N$ the operator norm of $X_N$
w.r.t. the measure $P_{\tau}$ it is determined to leading order by the mean of $X_N$, if the mean does not vanish.
Therefore the operator norm scales with $N$, except for the special case that the matrix entries are
$\IP^{\,\mu,\Lambda}$-almost surely centered. We prove in addition in \cite{KK2} that the $N$-scaling of the norm is due
to a single outlier of the spectrum by showing that the second largest eigenvalue (in modulus) possesses a $\sqrt{N}$-scaling
that is consistent with the law for the limiting spectral measure.

In \cite{KK2} we also generalize the just mentioned results to band matrices. Here
an additional difficulty arises,
because the mean of $X_N$ is no longer a multiple of $\mathcal{E}_N$ and will have large rank.
Nevertheless it is shown  that all results obtained for full matrices can be saved,
except for the result on the second largest eigenvalue (in modulus).

\end{document}